\documentclass[11pt,a4paper,english]{amsart}
\usepackage[utf8]{inputenc}
\usepackage[english]{babel} 

\usepackage{amsmath} 
\usepackage{amsthm} 
\usepackage{graphicx} 
\usepackage{xcolor} 
\usepackage{amsfonts}
\usepackage[biblabel]{cite}
\usepackage{bm} 
\usepackage[hidelinks]{hyperref} 
\usepackage{amssymb} 
\usepackage{tikz-cd} 
\usepackage{amscd}
\usepackage{etoolbox}
\patchcmd{\thmhead}{(#3)}{#3}{}{}
\usepackage{enumitem}
\usepackage{breqn} 

\DeclareMathOperator{\Supp}{Supp} 
\DeclareMathOperator{\ev}{ev} 

\DeclareMathOperator{\PRM}{PRM}
\DeclareMathOperator{\RM}{RM}

\newcommand{\F}{{\mathbb{F}}}
\newcommand{\fq}{\mathbb{F}_q}
\newcommand{\fqs}{\mathbb{F}_{q^s}}
\newcommand{\PP}{{\mathbb{P}}}
\newcommand{\PM}{{\mathbb{P}^{m}}}

\newcommand{\B}{{\mathcal{B}}}

\newcommand{\BD}{{\mathcal{B}_D}}

\usepackage{mathtools} 
\DeclarePairedDelimiter\abs{\lvert}{\rvert}%
\DeclarePairedDelimiter\norm{\lVert}{\rVert}%

\makeatletter
\let\oldabs\abs
\def\abs{\@ifstar{\oldabs}{\oldabs*}}
\let\oldnorm\norm
\def\norm{\@ifstar{\oldnorm}{\oldnorm*}}
\makeatother

\newtheorem{thm}{Theorem}[section]
\newtheorem{prop}[thm]{Proposition}
\newtheorem{cor}[thm]{Corollary}
\newtheorem{lem}[thm]{Lemma}
\theoremstyle{definition}

\newtheorem{rem}[thm]{Remark} 
 
\newtheorem{ex}[thm]{Example}

\title[A recursive construction for projective Reed-Muller codes]{A recursive construction for projective Reed-Muller codes}
\author{Rodrigo San-José}
\curraddr{
\texttt{Rodrigo San-José:} IMUVA-Mathematics Research Institute, Universidad de Valladolid, 47011 Valladolid (Spain).
}
\email{rodrigo.san-jose@uva.es}

\thanks{This work was supported in part by the following grants: Grant TED2021-130358B-I00 funded by MICIU/AEI/10.13039/501100011033 and by the ``European Union NextGenerationEU/PRTR'', Grant PID2022-138906NB-C21 funded by MICIU/AEI/ 10.13039/501100011033 and by ERDF/EU, and FPU20/01311 funded by the Spanish Ministry of Universities.}

\subjclass[2020]{Primary: 94B05. Secondary: 11T71, 14G50}

\keywords{Projective Reed-Muller codes, recursive construction, subfield subcodes, generalized Hamming weights}

\begin{document}

\maketitle

\begin{abstract}
We give a recursive construction for projective Reed-Muller codes in terms of affine Reed-Muller codes and projective Reed-Muller codes in fewer variables. From this construction, we obtain the dimension of the subfield subcodes of projective Reed-Muller codes for some particular degrees that give codes with good parameters. Moreover, from this recursive construction we derive a lower bound for the generalized Hamming weights of projective Reed-Muller codes which is sharp in most of the cases we have checked.
\end{abstract}

\section{Introduction}
Binary affine Reed-Muller codes can be constructed recursively via the $(u\mid u+v)$ construction, and, more generally, $q$-ary affine Reed-Muller codes can be constructed recursively using the matrix-product code construction \cite[Thm 5.6]{matrixproductcodes}. These recursive constructions provide a wealth of information about the code. For example, the recursive construction from \cite{matrixproductcodes} provides a simple proof for the minimum distance of affine Reed-Muller codes. Moreover, the subfield subcode of a code obtained using the $(u\mid u+v)$ construction can be obtained by applying the $(u\mid u+v)$ construction to the subfield subcodes of the component codes. In this work we are interested in a recursive construction for projective Reed-Muller codes, a generalization of affine Reed-Muller codes obtained by evaluating homogeneous polynomials in the projective space $\mathbb{P}^m$ which was introduced by Lachaud \cite{lachaudPRM} and whose basic parameters are presented in full generality in \cite{sorensen}. We apply this recursive construction to obtain information about the subfield subcodes and generalized Hamming weights of projective Reed-Muller codes.

Given a code $C\subset \fqs^n$, its subfield subcode with respect to the extension $\fqs\supset \fq$ is the linear code $C^\sigma:=C\cap \fq^n$. This is a standard procedure that has been used to construct long linear codes over a small finite field. In particular, this technique has been applied to obtain BCH codes as subfield subcodes of Reed-Solomon codes \cite{bierbrauercyclic}, and in the multivariate case, the subfield subcodes of $J$-affine variety codes (in particular, affine Reed-Muller codes) are well known and have been applied in several contexts \cite{galindo1,galindolcd,galindostabilizer}. The subfield subcodes of projective Reed-Solomon codes and projective Reed-Muller codes were studied in \cite{sanjoseSSCPRS} and \cite{sanjoseSSCPRM}, respectively. The primary challenge when dealing with subfield subcodes is the computation of a basis for the code, which, in particular, gives the dimension of the subfield subcode. However, one can check in \cite{sanjoseSSCPRM} that, for $m=2$, the expressions for the basis of the subfield subcode get quite involved, and for $m>2$, obtaining explicit expressions for the basis in general seems out of reach. In Section \ref{secssc}, we show that, for certain degrees, the recursive construction we obtain for projective Reed-Muller codes can be applied to their subfield subcodes as well. This directly gives the dimension of these subfield subcodes in a recursive manner, for any $m\geq 2$. Moreover, an explicit expression for a basis of these subfield subcodes can be obtained in terms of the basis from the subfield subcodes of affine Reed-Muller codes and the subfield subcodes of projective Reed-Muller codes with fewer variables. We also show that these particular degrees give codes with good parameters. 

With respect to the generalized Hamming weights of a code, these are a set of parameters that generalize the minimum distance of a code. One of the main applications of the generalized Hamming weights is that they characterize the performance of the code on the wire-tap channel of type II \cite{weiGHW}. The generalized Hamming weights of affine Reed-Muller codes were completely determined more than 20 years ago in \cite{pellikaanGHWRM}. However, the computation of the generalized Hamming weights of projective Reed-Muller codes in general remains an open problem and only partial results are known \cite{boguslavsky,beelenGHWPRM2,dattaGHWPRM3}. In \cite{beelenGHWPRM}, many of the previous results and hypotheses are collected, and the authors obtain the generalized Hamming weights of projective Reed-Muller codes in some cases for degree $d<q^s$. In Section \ref{secghw}, we use the recursive construction from Section \ref{secrecursive} to give a recursive lower bound for the generalized Hamming weights of a projective Reed-Muller code of any degree. Moreover, we also provide an upper bound that gives us a criterion to ensure that the bound is sharp in many cases. In the particular case of $m=2$, we give a more explicit expression for these bounds. By considering the general properties of generalized Hamming weights and our bounds, we obtain the exact values of the generalized Hamming weights of projective Reed-Muller codes in many cases. 

\section{Preliminaries}
We consider the finite field $\fq$ of $q$ elements with characteristic $p$, and its degree $s$ extension $\F_{q^s}$, with $s\geq 1$. We consider the projective space $\PM$ over $\fqs$. We denote by $p_j$ the number of points in $\PP^j$, i.e., $p_j=\frac{q^{s(j+1)}-1}{q^s-1}=q^{sj}+q^{s(j-1)}+\cdots +1$. Throughout this work, we will fix representatives for the points of $\PM$: for each point $[Q]\in \PM$, we choose the representative whose first nonzero coordinate from the left is equal to 1. We will denote by $P^m$ the set of representatives that we have chosen (seen as points in the affine space $\mathbb{A}^{m+1}$). Therefore, we have the following decomposition
$$
P^{m}=\left(\{1\}\times \F_{q^s}^m\right) \cup \left(\{0\}\times \{1\}\times \F_{q^s}^{m-1}\right)\cup\cdots\cup \{(0,\dots,0,1)\}.
$$
Moreover, we also can obtain this recursively by noting that 
\begin{equation}\label{descomposicion}
P^m=\left(\{1\}\times \F_{q^s}^m\right) \cup \left( \{0\}\times P^{m-1}\right).
\end{equation}

We consider now the polynomial ring $S=\fqs [x_0,\dots,x_m]$. Let $n=\abs{P^m}=p_m$. We define the following evaluation map: 
$$
\ev:S \rightarrow \fqs^{n},\:\: f\mapsto \left(f(Q_1),\dots,f(Q_n)\right)_{Q_i \in P^m}.
$$
Let $d$ be a positive integer. If we consider $S_d\subset S$, the set of homogeneous polynomials of degree $d$, we have that $\ev(S_d)$ is the projective Reed-Muller code of degree $d$, which we will denote by $\PRM_d(q^s,m)$, or $\PRM_d(m)$ if there is no confusion about the field. For $m=1$, we obtain the projective Reed-Solomon codes (sometimes called doubly extended Reed-Solomon codes), which are MDS codes with parameters $[q^s+1,d+1,q^s-d+1]$. For a code $C\subset \F_q^n$, we will denote its minimum distance by $d_1(C)$ (using the notation for generalized Hamming weights that we will consider in Section \ref{secghw}). For the case $m\geq 2$, we have the following results from \cite{sorensen} about the parameters of projective Reed-Muller codes and their duality.

\begin{thm}\label{paramPRM}
The projective Reed-Muller code $\PRM_d(q^s,m)$, $1\leq d\leq m(q^s-1)$, is an $[n,k]$-code with 
$$
\begin{aligned}
&n=\frac{q^{s(m+1)}-1}{q^s-1},\\
&k=\sum_{t\equiv d\bmod q^s-1,0<t\leq d}\left( \sum_{j=0}^{m+1}(-1)^j\binom{m+1}{j}\binom{t-jq^s+m}{t-jq^s}   \right).\\
\end{aligned}
$$
For the minimum distance, we have
$$
d_1(\PRM_d(q^s,m))=(q^s-\mu)q^{s(m-\nu-1)}, \text{ where } \;
d-1=\nu (q^s-1)+\mu, \;0\leq \mu <q^s-1.
$$
\end{thm}

\begin{thm}\label{dualPRM}
Let $1\leq d\leq m(q^s-1)$ and let $d^\perp=m(q^s-1)-d$. Then
$$
\begin{aligned}
&\PRM_d^\perp(q^s,m)=\PRM_{d^\perp}(q^s,m) &\text{ if } d\not\equiv 0\bmod (q^s-1), \\
&\PRM_d^\perp(q^s,m)=\PRM_{d^\perp}(q^s,m)+\langle (1,\dots,1) \rangle &\text{ if } d\equiv 0\bmod (q^s-1).
\end{aligned}
$$
\end{thm}

Let $d>0$ and let $M_i=\{x_i^{\alpha_i}\cdots x_m^{\alpha_m},\abs{\alpha}=d,\alpha_i>0,0\leq \alpha_j\leq q^s-1, i<j\leq m\}$, for $i=0,1,\dots,m$, and $M=\bigcup_{i=0}^m M_i$. One can check that the classes of the monomials in $M$ form a basis for $S_d/I(\PP^m)_d\cong \PRM_d(m)$ (for example, see \cite{projectivefootprint}), where $I(\PP^m)$ is the vanishing ideal of $\PP^m$, i.e, the ideal generated by the homogeneous polynomials that vanish at all the points of $\PP^m$. This also implies that the image by the evaluation map of $M$ is a basis for $\PRM_d(m)$.

We will also need to use affine Reed-Muller codes, which we denote by $\RM_d(q^s,m)$, or simply $\RM_d(m)$ if there is no confusion about the field. We consider the evaluation map
$$
\ev_{\mathbb{A}}:R \rightarrow \fqs^{q^{ms}},\:\: f\mapsto \left(f(Q_1),\dots,f(Q_n)\right)_{Q_i \in \fqs^m},
$$
where $R=\fqs[x_1,\dots,x_m]$. Let $R_{\leq d}$ be the polynomials of $R$ with degree less than or equal to $d$. Then we have $\RM_d(m):=\ev_{\mathbb{A}}(R_{\leq d})$. The following result about the parameters of affine Reed-Muller codes appears in \cite{delsarteRM,kasamiRM}.

\begin{thm}\label{paramRM}
The Reed-Muller code $\RM_d(q^s,m)$, $1\leq d\leq m(q^s-1)$, is an $[n,k]$-code with 
$$
\begin{aligned}
&n=q^{sm},\\
&k=\sum_{t=0}^d \sum_{j=0}^m (-1)^j\binom{m}{j}\binom{t-jq^s+m-1}{t-jq^s}.\\
\end{aligned}
$$
For the minimum distance, we have
$$
d_1(\RM_d(q^s,m))=(q^s-\mu)q^{s(m-\nu-1)}, \text{ where } \;
d=\nu(q^s-1)+\mu, \;0\leq \mu <q^s-1.
$$
\end{thm}

\section{A recursive construction for projective Reed-Muller codes}\label{secrecursive}
In this section, we introduce a recursive construction for projective Reed-Muller codes from affine Reed-Muller codes and projective Reed-Muller codes over a smaller projective space. The inspiration behind this idea was the fact that affine Reed-Muller codes can be obtained recursively by using matrix-product codes \cite[Thm. 5.6]{matrixproductcodes}.

To state this recursive construction, we need to consider a specific ordering of the points in $P^m$: we will assume that the first $q^{sm}$ points of $P^m$ are the points of $(\{1\}\times \fqs^m)$ (see decomposition (\ref{descomposicion})). Therefore, we have that the first $q^{sm}$ coordinates of the evaluation of a polynomial in $P^m$ correspond to the evaluation at the points $(\{1\}\times \fqs^m)$, and the rest of coordinates correspond to the evaluation at the points in the second part of the decomposition (\ref{descomposicion}). Let $\xi\in \F_{q^s}$ be a primitive element. We are also going to consider the following decomposition 
\begin{equation}\label{descomposicion2}
\F_{q^s}^m=P^{m-1}\cup \xi\cdot P^{m-1}\cup\cdots\cup \xi^{q^s-2}\cdot P^{m-1}\cup \{(0,\dots,0)\}.
\end{equation}
This decomposition is obtained by noting the following: given a point $Q$ in $\F_{q^s}^m\setminus \{(0,\dots,0\}$, its first nonzero coordinate is equal to $\xi^r$ for some $0\leq r\leq q^s-2$, which implies that $Q\in \xi^r \cdot P^{m-1}$. Therefore, using (\ref{descomposicion}) and (\ref{descomposicion2}) we have
\begin{equation}\label{descomposicion3}
P^m=\left(\{1\}\times \left(P^{m-1}\cup \xi\cdot P^{m-1}\cup\cdots\cup \xi^{q^s-2}\cdot P^{m-1}\cup \{(0,\dots,0)\}\right)\right) \cup \left( \{0\}\times P^{m-1}\right).
\end{equation}
We fix an ordering $\{Q'_1,\dots,Q'_{p_{m-1}}\}$ of the points of $P^{m-1}$. We assume that the first $p_{m-1}$ coordinates of the image of the evaluation map correspond to the evaluation at the points $\{1\}\times P^{m-1}$ (with the fixed ordering for $P^{m-1}$), the following $p_{m-1}$ coordinates correspond to the evaluation at the points $\{1\}\times \xi \cdot P^{m-1}$, etc. We fix this for all the points in $\{1\}\times \fqs^m$, and for the rest of the coordinates, which correspond to the evaluation in $ \{0\}\times P^{m-1}$, we also assume that we are using the same fixed ordering for $P^{m-1}$. 

Hence, for a given ordering of the points of $P^{m-1}$, we fix the ordering of all the points of $P^m$ as shown above. In what follows, we assume that the projective Reed-Muller codes are obtained by using the evaluation map over $P^m$ with this ordering for the points, and the affine Reed-Muller codes are obtained by evaluating in $\fqs^m$, ordered in the same way that we ordered $\{1\}\times \fqs^m$. 

\begin{thm}\label{recursive}
Let $1\leq d\leq m(q^s-1)$ and let $\xi$ be a primitive element in $\F_{q^s}$. We have the following recursive construction:
$$
\PRM_d(m)=\{(u+v_{\xi,d},v)\mid  u\in \RM_{d-1}(m), v\in \PRM_d(m-1)\},
$$
where $v_{\xi,d}:=v\times \xi^d v\times \cdots\times \xi^{(q^s-2)d}v\times \{0\}  =(v,\xi^d v,\xi^{2d}v,\dots,\xi^{(q^s-2)d}v,0)$.
\end{thm}

\begin{proof}
Taking into account that $x_0$ divides all the monomials in $M_0$ and the decomposition (\ref{descomposicion}), it is clear that $\langle \ev (M_0) \rangle=\RM_{d-1}(m)\times \{0\}^{p_{m-1}}$. 

Now we consider a monomial $x^\alpha\in \bigcup_{i=1}^m M_i$. Its evaluation in $\{0\}\times P^{m-1}$ (the second part of the decomposition (\ref{descomposicion})) is in $\PRM_d(m-1)$ because $x^\alpha$ only involves the last $m$ variables. In fact, the evaluation of all the monomials in $\bigcup_{i=1}^m M_i$ at the points of $\{0\}\times P^{m-1}$ gives a basis for $\PRM_d(m-1)$. Now assume that the evaluation of $x^\alpha$ at the points $\{0\}\times P^{m-1}$ is $v$. Let $\xi \in \fqs$ be a primitive element. Then, because of the ordering that we have chosen and the decomposition (\ref{descomposicion2}), it is clear that the evaluation of $x^\alpha$ at the points $\{1\}\times \fqs^m$ is precisely $v_{\xi,d}$. We have obtained the evaluation of the monomial $x^\alpha$ at both parts of the decomposition (\ref{descomposicion}). Thus, $\ev(x^\alpha)=(v_{\xi,d},v)$. 

If we have $f\in k[x_0,\dots,x_m]_d/I(\PM)_d\cong \PRM_d(m)$, then it can be written as
$$
f=\sum_i\lambda_i b_i +\sum_i \gamma_i a_i, b_i\in M_0,a_i\in \bigcup_{i=1}^m M_i,\lambda_i,\gamma_i\in\F_{q^s}.
$$
We have seen that $\ev(b_i)=u\times \{0\}^{p_{m-1}}$, for some $u\in \RM_{d-1}(m)$, and $\ev(a_i)=(v_{\xi,d},v)$, for some $v\in \PRM_d(m-1)$. We finish the proof by considering the linearity of the evaluation map. 
\end{proof}

\begin{rem}\label{remuplusv}
For $d\equiv 0 \bmod q^s-1$, the construction is simpler and, to some extent, resembles the $(u\mid u +v)$ construction:
$$
\PRM_d(m)=\{(u+v_{\xi,d},v), u\in \RM_{d-1}(m), v\in \PRM_d(m-1)\},
$$
with $v_{\xi,d}=(v,v,\dots,v,0)$.
\end{rem}

With the construction from Theorem \ref{recursive} we can recover the dimension of projective Reed-Muller codes from Theorem \ref{paramPRM} in a different way, which was already noted in \cite[Lem. 9]{renteriaRM}.

\begin{cor}\label{corodimRM}
We have that
$$
\dim(\PRM_d(m))=\dim(\RM_{d-1}(m))+\dim(\PRM_d(m-1)).
$$
\end{cor}
\begin{proof}
Using the notation in Theorem \ref{recursive}, we have $\dim(\RM_{d-1}(m))$ linearly independent vectors corresponding to $v=0$, and we have $\dim(\PRM_d(m-1))$ linearly independent vectors corresponding to $u=0$. By the construction from Theorem \ref{recursive}, every codeword of $\PRM_d(m)$ can be obtained as the sum of a vector with $u=0$ and a vector with $v=0$, and we obtain the result.
\end{proof}

\section{Subfield subcodes of projective Reed-Muller codes}\label{secssc}

As we stated at the beginning of the previous section, in some cases the recursive construction from Theorem \ref{recursive} resembles the $(u\mid u+v)$ construction. It is not hard to check that, given two codes $C_1,C_2$, the subfield subcode of the resulting code after using the $(u\mid u+v)$ construction with $C_1$ and $C_2$ is equal to the code obtained by applying the $(u\mid u+v)$ construction to $C_1^\sigma$ and $C_2^\sigma$ (recall that $C^\sigma=C\cap \fq^n$). Therefore, one may wonder if we can use the construction from Theorem \ref{recursive} to obtain results about the subfield subcodes of projective Reed-Muller codes, or even a recursive construction for them, which is what we study in this section. In \cite{sanjoseSSCPRM} the subfield subcodes of projective Reed-Muller codes were studied, mainly for the case $m=2$. Our approach in this section can be applied recursively for any $m$, and for the case $m=2$ our method provides an easier way to obtain the basis of the subfield subcode for some degrees. We start with a result about the minimum distance and dimension of the subfield subcodes of projective Reed-Muller codes. 

\begin{cor}
Let $1\leq d\leq m(q^s-1)$. We have the following inequalities:
$$
\begin{aligned}
d_1(\RM&_{d-1}(m))\leq d_1(\PRM_d^\sigma(m))\leq d_1(\RM_{d-1}^\sigma(m)),\\
&\dim(\PRM_d^\sigma (m))\geq \dim(\RM_{d-1}^\sigma(m)),
\end{aligned}
$$
and the last inequality is strict if $\PRM_d^\sigma (m)$ is non-degenerate.
\end{cor}
\begin{proof}
We have that $d_1(\PRM_d(m))\leq d_1(\PRM_d^\sigma(m))$ because $\PRM_d^\sigma(m)\subset \PRM_d(m)$, and one can check that $d_1(\PRM_d(m))=d_1(\RM_{d-1}(m))$ using Theorem \ref{paramPRM} and Theorem \ref{paramRM}. For the other inequalities, by Theorem \ref{recursive} we obtain $\RM_{d-1}(m)\times \{0\}^{p_{m-1}}\subset \PRM_d(m)$, which implies $\RM^\sigma_{d-1}(m)\times \{0\}^{p_{m-1}}\subset \PRM^\sigma_d(m)$.
\end{proof}

In most of the non-degenerate cases we have the equality for the three minimum distances in the previous result, although the bound is not always sharp as it was seen in \cite{sanjoseSSCPRM}. In the non-degenerate case we have $\dim(\PRM_d^\sigma (m))>\dim(\RM_{d-1}^\sigma(m))$, which may also be true in many degenerate cases, as one can check for the case $m=2$ in \cite[Cor 3.41]{sanjoseSSCPRM}.

For some specific degrees we obtain a recursive construction for the subfield subcodes of projective Reed-Muller codes, which in turn allows us to give more precise results about the parameters. 

\begin{cor}\label{recursivessc}
Let $\xi\in \fqs$ be a primitive element. Let $m>1$ and let ${d_\lambda}=\lambda\frac{q^s-1}{q-1}$ for some $\lambda\in \{ 1,2,\dots,m(q-1)\}$. Then we have
$$
\PRM_{d_\lambda}^\sigma (m)=\{(u+v_{\xi,d_\lambda},v), u\in \RM_{d_\lambda-1}^\sigma(m), v\in \PRM_{d_\lambda}^\sigma(m-1)\}.
$$
As a consequence, we obtain:
$$
\dim(\PRM_{d_\lambda}^\sigma (m))= \dim(\RM_{d_\lambda-1}^\sigma(m))+\dim(\PRM_{d_\lambda}^\sigma(m-1)).
$$
\end{cor}
\begin{proof}
We have that
$$
(\xi^{d_\lambda})^{q-1}=\xi^{\lambda (q^s-1)}=1 \implies \xi^{d_\lambda}\in \F_q.
$$
Then it is clear that for any $u\in \RM_{d_\lambda-1}^\sigma(m)$, $v\in \PRM_{d_\lambda}^\sigma(m-1)$, we have $v_{\xi,d_\lambda}\in \fq^n$ and $(u+v_{\xi,d_\lambda},v)\in \PRM_{d_\lambda}^\sigma(m)$ because of Theorem \ref{recursive}.

On the other hand, if we have $w\in\PRM_{d_\lambda}^\sigma(m)$, by Theorem \ref{recursive} we know that $w$ is of the form $(u+v_{\xi,d_\lambda},v)$, for $u\in \RM_{d-1}(m), v\in \PRM_d(m-1)$, and its coordinates are in $\F_q$. Therefore, $v$ must have its coordinates in $\fq$, i.e., $v\in \PRM_d^\sigma(m-1)$. Moreover, taking into account that $\xi^{d_\lambda}\in \F_q$, we also get that $v_{\xi,d_\lambda}\in \fq^n$, which implies that $u \in \RM_{d-1}^\sigma(m)$. Arguing as in Corollary \ref{corodimRM} we obtain the formula for the dimension. 
\end{proof}

\begin{rem}
Note that the hypothesis about the degrees in Corollary \ref{recursivessc} is necessary. For instance, in \cite[Cor. 3.41]{sanjoseSSCPRM} it can be seen that the formula for the dimension from Corollary \ref{recursivessc} does not hold in general when $d\neq d_\lambda$ for any $\lambda\in \{ 1,2,\dots,m(q-1)\}$.
\end{rem}

In all the cases from Corollary \ref{recursivessc} we can obtain a set of polynomials such that their image by the evaluation map is a basis for $\PRM_{d_\lambda}^\sigma(m)$ in a straightforward manner. To do so, for a given degree $d>0$, we define the homogenization up to degree $d$ of a polynomial $f\in \fqs[x_1,\dots,x_m]$ with degree $\deg(f)<d$ as $f^h=x_0^df(x_1/x_0)\in \fqs[x_0,\dots,x_m]_d$. Note that, with this definition, we always have that $x_0$ divides $f^h$. 

\begin{cor}\label{recursivesscbase}
Let ${d_\lambda}=\lambda\frac{q^s-1}{q-1}$ for some $\lambda\in \{ 1,2,\dots,m(q-1)\}$. Let $\{f_i\}_i\subset \fqs[x_1,\dots,x_m]_{\leq d-1}$ (resp. $\{g_j\}_j\subset \fqs[x_1,\dots,x_m]_d$) be a set of polynomials such that their evaluation in $\fqs^m$ (resp. $P^{m-1}$) is a basis for $\RM_{d-1}^\sigma(m)$ (resp. $\PRM_d^\sigma(m-1)$). Then the image by the evaluation map of $\{f_i^h\}_i\cup \{g_j\}_j\subset \fqs[x_0,\dots,x_m]$ over $P^m$ is a basis for $\PRM_d^\sigma(m)$. 
\end{cor}
\begin{proof}
Assume that we have the sets $\{f_i\}_i$ and $\{g_j\}_j$ as in the statement. Then we have that $\{f_i^h\}_i$ is a set of homogeneous polynomials of degree $d$, and the image by the evaluation map of this set generates $\RM_{d-1}^\sigma(m)\times \{0\}^{p_{m-1}}$, i.e., the vectors with $v=0$ of the recursive construction from Corollary \ref{recursivessc}. 

Let $\xi\in \fqs$ be a primitive element. If the evaluation of a polynomial of $\{g_j\}_j$ over $P^{m-1}$ is $v\in \PRM^\sigma_d(m-1)$, then, when regarding this polynomial in $\fqs[x_0,\dots,x_m]$, the evaluation over $P^m$ is precisely $(v_{\xi,d},v)$ (this is the idea of the proof of Theorem \ref{recursive}). Therefore, the image by the evaluation map of $\{f_i^h\}_i\cup \{g_j\}_j$ over $P^m$ generates 
$$
\{(u+v_{\xi,d_\lambda},v), u\in \RM_{d_\lambda-1}^\sigma(m), v\in \PRM_{d_\lambda}^\sigma(m-1)\},
$$
which is $\PRM_d^\sigma(m)$ by Corollary \ref{recursivessc}.
\end{proof}

In both Corollary \ref{recursivessc} and Corollary \ref{recursivesscbase}, in order to use the recursion we need to obtain bases for affine and projective Reed-Muller codes for some $m\geq 1$. Sets of polynomials whose evaluation are a basis for the subfield subcodes of affine Reed-Muller codes are known for any number of variables \cite[Thm. 11]{galindolcd}, and for projective Reed-Muller codes we also know how to obtain such sets of polynomials in the case of projective Reed-Solomon codes \cite{sanjoseSSCPRS} and in the case of projective Reed-Muller codes over the projective plane \cite{sanjoseSSCPRM}. For $m>2$, we can apply Corollary \ref{recursivessc} and Corollary \ref{recursivesscbase} recursively until we reach the known cases of $m=1$ or $m=2$. 

\begin{ex}
Let $q^s=4$, $d=3$, $m=3$, and let $\xi$ be a primitive element of $\F_4$. We are going to obtain a set of polynomials such that its image by the evaluation map is a basis for $\PRM_d^\sigma(m)$. Using Corollary \ref{recursivesscbase}, we need to compute a basis for $\RM_{2}^\sigma(3)$ and $\PRM_3^\sigma(2)$. From \cite[Thm. 11]{galindolcd}, we obtain that the image by the evaluation map over $\F_4^3$ of the following polynomials is a basis for $\RM_2^\sigma(3)$:
$$
B_1=\{ 1,x_1+x_1^2,\xi x_1+\xi^2x_1^2,x_2+x_2^2,\xi x_2+\xi^2x_2^2,x_3+x_3^2,\xi x_3+\xi^2x_3^2 \}.
$$
From \cite{sanjoseSSCPRM}, the image by the evaluation map over $P^2$ of the polynomials 

$$
\begin{aligned}
B_2=\{ x_{1}^3, x_{1} x_{3}^{2} + x_{1}^2 x_{3}, (\xi + 1) x_{1} x_{3}^{2} + \xi x_{1}^2 x_{3}, x_{1} x_{2}^{2} + x_{1}^2 x_{2}, (\xi + 1) x_{1} x_{2}^{2} + \xi x_{1}^2 x_{2}, x_{2}^{3},\\x_{2}^{2} x_{3} + x_{2} x_{3}^{2}, 
(\xi + 1) x_{2}^{2} x_{3} + \xi x_{2} x_{3}^{2}, x_{3}^{3}\}
\end{aligned}
$$
forms a basis for $\PRM_3^\sigma(2)$. By Corollary \ref{recursivesscbase}, the image by the evaluation map over $P^3$ of $B_1^h\cup B_2$ is a basis for $\PRM_3^\sigma(3)$, where
$$
\begin{aligned}
B_1^h=\{ x_0^3,x_0^2x_1+x_0x_1^2,\xi x_0^2x_1+\xi^2x_0x_1^2,x_0^2x_2+x_0x_2^2,\xi x_0^2x_2+\xi^2x_0x_2^2,x_0^2x_3+x_0x_3^2,&\\
\xi x_0^2x_3+\xi^2x_0x_3^2 \}.&
\end{aligned}
$$
\end{ex}

As the dual of a projective Reed-Muller codes is another projective Reed-Muller code for $d\not\equiv 0\bmod q^s-1$ (by Theorem \ref{dualPRM}), the recursive construction from Theorem \ref{recursive} can be used for the dual codes, except for the case $d\equiv 0\bmod q^s-1$. However, it is not true in general that $\PRM_d^{\sigma,\perp}(m):=(\PRM_d^{\sigma}(m))^\perp$ is equal to $\PRM_d^{\perp,\sigma}(m):=(\PRM_d^{\perp}(m))^\sigma$. Therefore, the construction from Corollary \ref{recursivessc}  does not apply to the dual code of the subfield subcode of a projective Reed-Muller code. Nevertheless, in the next result we show that the dual codes can also be obtained from a similar recursive construction.

\begin{prop}\label{recursivedualssc}
Let $\xi\in \fqs$ be a primitive element. Let $m>1$ and let ${d_\lambda}=\lambda\frac{q^s-1}{q-1}$ for some $\lambda\in \{ 1,2,\dots,m(q-1)\}$. Then we have
$$
\PRM_{d_\lambda}^{\sigma,\perp} (m)=\{(u^t,v^t-u_{\xi,d}^t), u^t\in \RM_{d_\lambda-1}^{\sigma,\perp}(m), v^t\in \PRM_{d_\lambda}^{\sigma,\perp}(m-1)\},
$$
where, taking into account the decomposition (\ref{descomposicion2}), if we have $u^t=(u^t_0,u^t_1,\dots,u^t_{q^s-2},u^t_{q^s-1})$, with $u^t_0,\dots,u^t_{q^s-2}\in \F_q^{p_{m-1}}$ and $u^t_{q^s-1}\in \F_q$, then 
$$
u_{\xi,d}^t:=u_0^t+\xi^d u_1^t+\cdots + \xi^{(q^s-2)d}u_{q^s-2}^t=\sum_{i=0}^{q^s-2}\xi^{i\cdot d} u_i^t. 
$$
\end{prop}
\begin{proof}
The vector space $\{(u^t,v^t-u_{\xi,d}^t), u^t\in \RM_{d_\lambda-1}^{\sigma,\perp}(m), v^t\in \PRM_{d_\lambda}^{\sigma,\perp}(m-1)\}$ has dimension $\dim \RM_{d_\lambda-1}^{\sigma,\perp}(m)+\dim \PRM_{d_\lambda}^{\sigma,\perp}(m-1)$, which is the dimension of $\PRM_{d_\lambda}^{\sigma,\perp}(m)$ according to Corollary \ref{recursivessc}. Therefore, if we consider $u\in \RM_{d_\lambda-1}^\sigma(m)$, $v\in \PRM_{d_\lambda}^{\sigma}(m-1)$, $u^t\in \RM_{d_\lambda-1}^{\sigma,\perp}(m)$ and $v^t\in \PRM_{d_\lambda}^{\sigma,\perp}(m-1)$, by Corollary \ref{recursivessc} we just need to verify that 
$$
(u+v_{\xi,d},v)\cdot (u^t,v^t-u_{\xi,d}^t)= v_{\xi,d}\cdot u^t -v\cdot u^t_\xi=0.  
$$
By considering the decomposition (\ref{descomposicion2}), we can divide the vector $u^t$ as in the statement of this result, and the previous expression can be written as
$$
v_{\xi,d}\cdot u^t -v\cdot u^t_\xi=v\cdot u^t_0+ \xi^d v\cdot u^t_1+\cdots+\xi^{(q^s-2)d}v\cdot u^t_{q^s-2}-v\cdot (u_0^t+\xi^d u_1^t+\cdots + \xi^{(q^s-2)d}u_{q^s-2}^t),
$$
which is equal to 0. 
\end{proof}

As with the recursive construction from Corollary \ref{recursivessc}, the previous result can be used recursively because we know bases for $\RM_{d_\lambda-1}^{\sigma,\perp}(m)$ (for example, see \cite{galindostabilizer}), and also for $\PRM_{d_\lambda}^{\sigma,\perp}(m)$, for $m=1$ and $m=2$, see \cite{sanjoseSSCPRM,sanjoseSSCPRS}.

\subsection{Examples}
In this subsection we show that we can obtain good parameters with the subfield subcodes of projective Reed-Muller codes appearing in Corollary \ref{recursivessc} and Proposition \ref{recursivedualssc}. For $m=2$, the codes arising from Corollary \ref{recursivessc} are a particular case of the codes studied in \cite{sanjoseSSCPRM}. However, for the degrees considered in Corollary \ref{recursivessc} we have an easier construction, and we show in Table \ref{tablagoodparam} that many of the codes with good parameters from \cite{sanjoseSSCPRM} correspond precisely to the codes from Corollary \ref{recursivessc} or their duals from Proposition \ref{recursivedualssc}. For the minimum distance of the codes from Corollary \ref{recursivessc} we can use the bound $d_1(\PRM_d^\sigma(m))\geq d_1(\PRM_d(m))$, and for the duals we can compute the minimum distance with Magma \cite{magma}. All codes presented in Table \ref{tablagoodparam} exceed the Gilbert-Varshamov bound, and some of them have the best known parameters according to codetables.de \cite{codetables}, as stated in \cite{sanjoseSSCPRM}. 

\begin{table}[ht]
\caption{Codes with good parameters appearing in \cite{sanjoseSSCPRM} that can be obtained from Corollary \ref{recursivessc} or Proposition \ref{recursivedualssc}.} 
\label{tablagoodparam}
\centering
\begin{tabular}{||c|c|c|c|c||c|c|c||}
 \hline 
 $q$ & $s$ & $m$ &$\lambda $ & Result &$n$ & $k$ & $d_1(C) \geq $  \\
  \hline \hline
2&2& 2& 1& \ref{recursivessc} & 21&9&8\\
2&2&2& 1& \ref{recursivedualssc} & 21&12&5\\
2&2&3&1&\ref{recursivessc} &85&16&32\\
2&2&3&2&\ref{recursivedualssc} &85&25&21\\
2&2&3&2&\ref{recursivessc} &85&60&8\\
2&2&3&1&\ref{recursivedualssc} &85&69&5\\
3&9&2&1&\ref{recursivessc} &91&9&54\\
3&9&2&1&\ref{recursivedualssc} &91&82&4\\
4 & 2&2& 1&\ref{recursivessc} & 273 & 9 & 192 \\
5 & 2 &2& 1 &\ref{recursivessc} & 651 & 9 & 500\\
7 & 2 &2& 1 & \ref{recursivessc} &2451 & 9 & 2058 \\
\hline
\end{tabular}
\end{table}

Furthermore, as the recursive approach from this work allows us to work easily with $m>2$, we can also provide examples of good parameters which do not appear in \cite{sanjoseSSCPRM}. As these codes are very long when we increase $q^s$, we only provide a few examples that still have moderate lengths. For the extension $\F_4\supset \F_2$ and $m=4$, for $\lambda=1,2,$ from Corollary \ref{recursivessc} we obtain the codes $[341,25,128]_2$ and $[341,295,8]_2$, respectively, which surpass the Gilbert-Varshamov bound. The dual of the code with parameters $[341,25,128]_2$ has parameters $[341,316,5]_2$, which also exceed the Gilbert-Varshamov bound. 

For the extension $\F_9\supset \F_3$ and $m=3$, we can consider $\lambda=1$ in Corollary \ref{recursivessc}, which gives a code with parameters $[820,16,486]_3$. Its dual has parameters $[820,804,4]_3$, and both of them surpass the Gilbert-Varshamov bound. We also note that, as we claimed in the introduction, by considering subfield subcodes we are able to obtain long codes with good parameters over a small finite field.

\section{A bound for the generalized Hamming weights of projective Reed-Muller codes}\label{secghw}
In this section we provide a lower bound for the generalized Hamming weights of any projective Reed-Muller code. Let $C\subset \fqs^n$ be a linear code and $D\subset C$. The support of $D$, denoted by $\Supp(D)$, is defined as
$$
\Supp(D):=\{i\mid \exists\ c=(c_1,\dots,c_n)\in D,\; c_i\neq 0 \}.
$$
If $D$ is a linear subspace contained in $C$, then we say that it is a subcode of $C$. The $r$th generalized Hamming weight of $C$, denoted by $d_r(C)$, is defined as
$$
d_r(C)=\min \{ \abs{\Supp(D)}\mid D \text{ is a subcode of $C$ with } \dim D=r\}.
$$

\begin{rem}\label{rembaseGHW}
Given a basis $B=\{b_1,\dots,b_k\}$ for a subcode $D$, we have that
$$
\Supp(D)=\bigcup_{i=1}^k \Supp(b_i).
$$
\end{rem}

The generalized Hamming weights satisfy a Singleton-type bound and they are monotonous, as it is shown in the following results from \cite{weiGHW}.

\begin{thm}[(Monotonicity)]\label{monotonia}
For an $[n,k]$ linear code $C$ with $k>0$ we have
$$
1\leq d_1(C)<d_2(C)<\cdots <d_k(C)\leq n.
$$
\end{thm}
\begin{cor}[(Generalized Singleton Bound)]\label{singletongeneralizada}
For an $[n,k]$ linear code $C$ we have
$$
d_r(C)\leq n-k+r, \; 1\leq r\leq k.
$$
\end{cor}

For an MDS code $C$ with length $n$ and dimension $k$, Theorem \ref{monotonia} and Corollary \ref{singletongeneralizada} imply that 
$$
d_r(C)=n-k+r, \;1\leq r\leq k.
$$

For the affine case, the generalized Hamming weights of Reed-Muller codes were obtained in \cite{pellikaanGHWRM}, where the authors give several ways to compute them. We present one of them now, which does not require additional machinery, see \cite[Thm. 5.10]{pellikaanGHWRM}. Let $Q=\{0,\dots,q^s-1\}$. We consider the set $Q^m$ with the lexicographic order defined by 
$$
(\alpha_1,\dots,\alpha_m)\prec (\beta_1,\dots,\beta_m) \iff \alpha_1=\beta_1,\dots,\alpha_\ell=\beta_\ell, \;\alpha_{\ell+1}<\beta_{\ell+1}, \textnormal{ for some }\ell.
$$
For $\beta \in Q^m$,  we denote $\deg(\beta)=\sum_{i=1}^m\beta_i$.

\begin{thm}\label{ghwRM}
Let $\beta$ be the $r$th element of $Q^m$ in the previous lexicographic order with the property that 
$$
\deg(\beta)>m(q^s-1)-d-1.
$$
Then
$$
d_r(\RM_d(m))=\sum_{i=1}^m \beta_{m-i+1}q^{s(i-1)}+1.
$$
\end{thm}

Our goal in this section is to provide a general lower bound for the generalized Hamming weights of any projective Reed-Muller code using the construction from Theorem \ref{recursive}. We do this in the next result, where we understand that $\PRM_d=\{0\}$ for $d\leq 0$, and, as before, $d_0(C)=0$ for a code $C\subset \fqs^n$. 

\begin{thm}\label{cotaghw}
Let $1\leq d\leq m(q^s-1)$ and $2\leq r\leq \dim(\PRM_d(m))$. We consider
$$
    Y = \left\{(\alpha,\gamma): \begin{array}{c}
        \max\{r-\dim \RM_{d-1}(m),0\} \leq  \alpha \leq \min \{ \dim \PRM_d(m-1),r\}\\
        \max\{r-\dim\RM_d(m),0\}\leq \gamma \leq  \min\{\dim \PRM_{d-(q^s-1)}(m-1),\alpha\} \\
        \end{array}\right\}.
$$
Then we have
$$
d_r(\PRM_d(m))\geq \min_{(\alpha,\gamma)\in Y}B_{\alpha,\gamma},
$$
where 
$$
\begin{aligned}
B_{\alpha,\gamma}:=& \max(d_{r-\gamma}(\RM_d(m)),d_{r-\alpha}(\RM_{d-1}(m))) \\
&+\max(d_{\alpha}(\PRM_d(m-1)),d_\gamma(\PRM_{d-(q^s-1)}(m-1))).
\end{aligned}
$$
\end{thm}
\begin{proof}
We fix a degree $d$, and let $\xi\in \fqs$ be a primitive element. Let $k_u=\dim(\RM_{d-1}(m))$ and $k_v=\dim(\PRM_d(m-1))$. We consider a basis $\{u^i\}_{i=1}^{k_u}$ for $\RM_{d-1}(m)$ and a basis $\{v^j\}_{j=1}^{k_v}$ for $\PRM_d(m-1)$. Then, by Theorem \ref{recursive}, we have that $\mathcal{B}=\{u^i\times \{0\}^{p_{m-1}}\}_{i=1}^{k_u}\cup \{(v^j_{\xi,d},v^j)\}_{j=1}^{k_v}$ is a basis for $\PRM_d(m)$. Let $D$ be a subcode of $\PRM_d(m)$ with $\dim D=r$. We consider a basis
$$
\mathcal{B}_D=\{b_l,1\leq l\leq r\}:=\left\{\sum_{i=1}^{k_u} \lambda_{l,i} u^i \times \{0\}^{p_{m-1}} +\sum_{j=1}^{k_v}\mu_{l,j} (v^j_{\xi,d},v^j) , 1\leq l \leq r\right\}.
$$
Now we divide the vectors $b_l$ into two parts, $b_l=(b_{l,1},b_{l,2})$, where 
\begin{equation}\label{defbl}
b_{l,1}:=\sum_{i=1}^{k_u}\lambda_{l,i} u^i+\sum_{j=1}^{k_v}\mu_{l,j}v^j_{\xi,d} \;,\; b_{l,2}:=\sum_{j=1}^{k_v}\mu_{l,j}v^j. 
\end{equation}

By the definition, it is clear that $b_{l,2}\in \PRM_d(m-1)$. On the other hand, if $b_{l,2}=0$, then $\mu_{l,j}=0$ for all $j$ and $b_{l,1}\in \RM_{d-1}(m)$. Moreover, in general we have $b_{l,1}\in \RM_d(m)$. This is because $b_{l,1}$ is the first part of $b_l$, which is a vector from $\PRM_d(m)$, i.e., is the evaluation of a homogeneous polynomial $f$ of degree $d$. The evaluation of $f$ in the first part of the decomposition (\ref{descomposicion}) is the same as the evaluation of $g:=f(1,x_1,\dots,x_m)$ in $\mathbb{A}^m$, which belongs to $\RM_d(m)$. 

If we consider the matrix $G$ whose rows are given by the vectors in $\BD$ we obtain
$$
G=\begin{pmatrix}
b_{1,1} & b_{1,2}\\
b_{2,1} & b_{2,2}\\
\vdots & \vdots \\
b_{r,1} & b_{r,2}\\
\end{pmatrix}
.
$$
By performing row operations in $G$, and reordering the $b_{l,i}$ if necessary, we can assume that there is an integer $\alpha\leq r$ such that the set $\{ b_{1,2},\dots,b_{\alpha,2}\}$ is linearly independent, and $b_{\alpha+1,2}=\cdots=b_{r,2}=0$. Therefore, we have
$$
G=\begin{pmatrix}
b_{1,1} & b_{1,2}\\
\vdots & \vdots \\
b_{\alpha,1}& b_{\alpha,2}\\
b_{\alpha+1,1}& 0 \\
\vdots & \vdots \\
b_{r,1} & 0\\
\end{pmatrix}.
$$
With the current ordering, note that the set $\{b_{\alpha+1,1},\dots,b_{r,1}\}$ is linearly independent (because $\mathcal{B}_D$ is a basis and $G$ is a full rank matrix). Therefore, we can perform row operations in such a way that, after reordering the $b_{l,i}$ (if necessary), for $1\leq l\leq \alpha$, we have an integer $\gamma\leq \alpha$ such that $b_{1,1}=\cdots=b_{\gamma,1}=0$, and the set $\{b_{\gamma+1,1},\dots,b_{r,1}\}$ is linearly independent. Therefore, we can assume that $G$ has the form 
\begin{equation}\label{formaG}
G=\begin{pmatrix}
0 & b_{1,2}\\
\vdots & \vdots \\
0 & b_{\gamma,2} \\
b_{\gamma+1,1} & b_{\gamma+1,2}\\
\vdots & \vdots \\
b_{\alpha,1}& b_{\alpha,2}\\
b_{\alpha+1,1}& 0 \\
\vdots & \vdots \\
b_{r,1} & 0\\
\end{pmatrix}=:\begin{pmatrix}G_1 &G_2
\end{pmatrix},
\end{equation}
where $\{b_{1,2},\dots,b_{\alpha,2}\}$ and $\{b_{\gamma+1,1},\dots,b_{r,1}\}$ are linearly independent sets. Now we will give a lower bound for $\abs{\Supp(D)}$ for any subcode $D$ of $\PRM_d(m)$ depending on the values of $\alpha$ and $\gamma$. Note that these values do not depend on the choice of $\B_D$. 

Assuming that $G$ has the form from (\ref{formaG}), by the reasoning after equation (\ref{defbl}), we see that $\{b_{\alpha+1,1},\dots,b_{r,1}\}$ is contained in $\RM_{d-1}(m)$ and $\{b_{1,2},\dots,b_{\alpha,2}\}$ is contained in $\PRM_d(m-1)$. Both of these sets are linearly independent because of the assumptions on the form of $G$. Therefore, $r- \dim \RM_{d-1}(m)\leq \alpha \leq \dim \PRM_d(m-1)$ (we also have the obvious condition $\alpha\leq r$). 

In order to bound $\abs{\Supp(D)}$, we note that $\Supp(D)$ is the union of the supports of $b_l$, $1\leq l \leq r$, by Remark \ref{rembaseGHW}. Therefore, it is enough to study the union of the supports of the rows of $G$. Moreover, we can study the union of the support of the $b_{i,1}$ and $b_{j,2}$, $1\leq i,j \leq r$, separately, which corresponds to studying the union of the supports of the rows of $G_1$ and $G_2$, which we denote by $\Supp(G_1)$ and $\Supp(G_2)$, respectively. 

For $G_1$, by considering the last $r-\alpha$ rows it is clear that $\abs{\Supp(G_1)}\geq d_{r-\alpha}(\RM_{d-1}(m))$ (recall that $b_{l,1}\in \RM_{d-1}(m)$ if $b_{l,2}=0$). Another possible bound is $\abs{\Supp(G_1)}\geq d_{r-\gamma}(\RM_{d}(m))$ (see the reasoning after equation \ref{defbl}). Therefore,
$$
\abs{\Supp(G_1)}\geq \max({d_{r-\alpha}(\RM_{d-1}(m)),d_{r-\gamma}(\RM_{d}(m))}).
$$
For $G_2$, since $b_{l,2}\in \PRM_d(m-1)$, we have $\abs{\Supp(G_2)}\geq d_\alpha(\PRM_d(m-1))$. This bound can be improved by studying the first $\gamma$ rows of $G$. This is because in order to have $b_{l,1}=0$, by Theorem \ref{recursive} we also must have $b_{l,1}=u+v_{\xi,d}=0$ for some $u\in \RM_{d-1}(m)$, $v\in \PRM_d(m-1)$. The vector $u$ is the evaluation of a polynomial $f$ of degree at most $d-1$ in $\mathbb{A}^m$, and $v_{\xi,d}$ is the evaluation of a homogeneous polynomial $g$ of degree $d$ in $\mathbb{A}^m$.
As $f$ and $g$ have the opposite evaluation in $\mathbb{A}^m$, we must have $f\equiv -g \bmod \langle x_1^{q^s}-x_1,\dots,x_{m}^{q^s}-x_{m}\rangle$. This implies that, if we consider $\overline{f}$ and $\overline{g}$ the polynomials obtained by reducing all the exponents of the monomials of $f$ and $g$, respectively, modulo $q^s-1$, then $\overline{f}=-\overline{g}$. As $f$ is of degree at most $d-1$, $\overline{f}$ and $\overline{g}$ are of degree at most $d-1$. Taking into account that $g$ is homogeneous of degree $d$, we deduce that all the exponents of the monomials of $g$ can be reduced modulo $q^s-1$ (in order to have $\overline{g}$ of degree at most $d-1$), which implies that all the monomials from $g$ reduce to monomials of degree at most $d-(q^s-1)$ in $\overline{g}$. Hence, $g$ has the same evaluation as some homogeneous polynomial of degree $d-(q^s-1)$. Thus, $v\in \PRM_{d-(q^s-1)}(m-1)$. What we have obtained is that $b_{l,2}\in\PRM_{d-(q^s-1)}(m-1)$, $1\leq l\leq \gamma$, and 
$$
\abs{\Supp(G_2)}\geq \max(d_{\alpha}(\PRM_d(m-1)),d_\gamma(\PRM_{d-(q^s-1)}(m-1))).
$$
Therefore, for $D$ we have
\begin{equation}\label{defcota}
\begin{aligned}
\abs{\Supp(D)} \geq B_{\alpha,\gamma}=& \max(d_{r-\gamma}(\RM_d(m)),d_{r-\alpha}(\RM_{d-1}(m))) \\
&+\max(d_{\alpha}(\PRM_d(m-1)),d_\gamma(\PRM_{d-(q^s-1)}(m-1))).
\end{aligned}
\end{equation}
Note that, because of the previous reasoning and the form of $G$, we must have $r-\dim \RM_d(m)\leq \gamma \leq \dim \PRM_{d-(q^s-1)}(m-1)$ (besides $\gamma\leq \alpha$). Since any subcode $D$ gives some $(\alpha,\gamma)\in Y$ by the previous reasoning, and taking (\ref{defcota}) into account, we obtain the result.
\end{proof}

\begin{rem}\label{remdmenorq}
Note that we can only have $\gamma>0$ in Theorem \ref{cotaghw} if $d\geq q^s$.
\end{rem}

With the previous result we obtain lower bounds for all the generalized Hamming weights of projective Reed-Muller codes recursively. This is because we can apply the previous theorem recursively until we get to projective Reed-Solomon codes, in which case we have $d_r(\PRM_d(1))=\max\{q^s-d+r,r\}$, for $1\leq r\leq d+1$. With this starting point, and using Theorem \ref{ghwRM} for the generalized Hamming weights of affine Reed-Muller, we can compute the value of the previous bound for the generalized Hamming weight of any projective Reed-Muller code. This contrasts with the results from \cite{beelenGHWPRM}, where the authors compute the generalized Hamming weights of projective Reed-Muller codes in some cases and they also provide lower bounds, but only for $d< q^s$, while our bound works for any degree and any $r$. 

We have checked by computer that the bound in Theorem \ref{cotaghw} is sharp for all the generalized Hamming weights with $m=2,3,$ and $q^s=2$; $m=2$ and $q^s=3$; and also for some particular degrees, such that $m=3$, $q^s=3$ and $d=1$; and $m=2$, $q^s=4$ and $d=1,2$. As computing the generalized Hamming weights of a code is computationally intensive, we can only do it for smaller examples. It is therefore desirable to have some criterion to guarantee that the bound from Theorem \ref{cotaghw} is sharp in some cases, which is the purpose of the following result. 

\begin{lem}\label{lemasup}
Let $1\leq d\leq m(q^s-1)$ and $2\leq r\leq \max\{\dim \RM_{d-1}(m),\dim \PRM_d(m-1) \}$. Then
$$
d_r(\PRM_d(m))\leq \min\{ d_r(\RM_{d-1}(m)),q^s\cdot d_r(\PRM_d(m-1))\},
$$
where if $r> \dim\RM_{d-1}(m)$ or $r> \dim \PRM_d(m-1)$ we do not consider the corresponding generalized Hamming weight in the minimum.
\end{lem}
\begin{proof}
We are going to find subcodes $E'$ of $\PRM_d(m)$ such that $\abs{\Supp(E')}=d_r(\RM_{d-1}(m))$ or $\abs{\Supp(E')}=q^s\cdot d_r(\PRM_d(m-1))$. Assuming $r\leq \dim\RM_{d-1}(m)$, there is a subcode $E\subset \RM_{d-1}(m)$ with $\dim E=r$ and $\abs{\Supp(E)}=d_r(\RM_{d-1}(m))$. Then, the subcode $E'=E\times \{0\}^{p_{m-1}}\subset \PRM_d(m)$ verifies $\abs{\Supp(E')}=d_r(\RM_{d-1}(m))$.

If we assume that $r \leq \dim \PRM_d(m-1)$ instead, there is a subcode $E\subset \PRM_d(m-1)$ with $\dim E=r$, such that $\abs{\Supp(E)}=d_r(\PRM_d(m-1))$. The subcode 
$$
E':=\{(v_{\xi,d},v), v\in E \}\subset \PRM_d(m)
$$
verifies $\abs{\Supp(E')}=q^s\cdot d_r(\PRM_d(m-1))$.
\end{proof}

We can use Lemma \ref{lemasup}, together with Theorem \ref{monotonia}, to ensure that the bound from Theorem \ref{cotaghw} is sharp in many cases. We see this in Example \ref{exnormal} and in Subsection \ref{secex}, where we show the tables that we obtain using our results for the generalized Hamming weights of projective Reed-Muller codes for several finite fields. Notwithstanding the foregoing, as we will see in Example \ref{exnosharp}, we can also find particular cases in which the bound is not sharp.

\begin{ex}\label{exnormal}
Let $q^s=4$, $d=5$, $m=2$ and $r=2$. In this example we compute the bound from Theorem \ref{cotaghw} for $d_2(\PRM_5(2))$. From Theorem \ref{paramRM} we can obtain $\dim \RM_4(2)=13$ and $\dim \RM_5(2)=15$. For $m=2$, $\PRM_d(m-1)$ is a projective Reed-Solomon code, and for $d=5>4=q^s$ we have $\PRM_5(1)=\fqs^{q^s+1}$, which implies $\dim \PRM_5(1)=q^s+1$. Moreover, for $d-(q^s-1)=2$, we have $\dim \PRM_2(1)=3$. Therefore, we obtain
$$
    Y = \left\{(\alpha,\gamma): \begin{array}{c}
        0 \leq  \alpha \leq 2\\
        0\leq \gamma \leq  \alpha \\
        \end{array}\right\}=\{(0,0),(1,0),(1,1),(2,0),(2,1),(2,2)\}.
$$
Now we have to compute $\min_{(\alpha,\gamma)\in Y}B_{\alpha,\gamma}$. To do it, we need to obtain the generalized Hamming weights for affine Reed-Muller codes and projective Reed-Solomon codes. For projective Reed-Solomon codes, we have already stated that we have $d_r(\PRM_d(1))=\max\{q^s-d+r,r\}$, and for affine Reed-Muller codes we can use Theorem \ref{ghwRM} to obtain
\begin{table}[ht]
\centering
\begin{tabular}{||c||c|c||}
\hline
    $r$ & 1&2 \\
    \hline \hline
    $d_r(\RM_5(2))$ & 2&3\\
    $d_r(\RM_4(2))$ & 3&4\\
    \hline
\end{tabular}
\end{table}

With all of this we can compute all the values $B_{\alpha,\gamma}$, for $(\alpha,\gamma)\in Y$:
$$
\begin{aligned}
B_{0,0} &=d_2(\RM_4(2))=4, \\
B_{1,0} &=\max \{ d_2(\RM_5(2)),d_1(\RM_4(2))\}+d_1(\PRM_5(1))=\max \{ 3,3\}+1=4,\\
B_{1,1} &=\max\{d_1(\RM_5(2)),d_1(\RM_4(2))\}+\max\{d_1(\PRM_5(1)),d_1(\PRM_2(1))\}=6,\\
B_{2,0} &=d_2(\RM_5(2))+d_2(\PRM_5(1))=3+2=5,\\
B_{2,1} &=d_1(\RM_5(2))+\max\{d_2(\PRM_5(1)),d_1(\PRM_2(1))\}=2+\max\{2,3\}=5,\\
B_{2,2} &=\max\{d_2(\PRM_5(1)),d_2(\PRM_2(1))\}=\max\{2,4\}=4.\\
\end{aligned}
$$
Therefore, we have obtained
$$
d_r(\PRM_d(m))\geq \min_{(\alpha,\gamma)\in Y}B_{\alpha,\gamma}=4.
$$
Moreover, as $\min_{(\alpha,\gamma)\in Y}B_{\alpha,\gamma}=B_{0,0}=d_r(\RM_{d-1}(2))$, by Lemma \ref{lemasup} we know that we have the equality $d_r(\PRM_d(m))=4$. 
\end{ex}

\begin{ex}\label{exnosharp}
In this example we show a particular case in which the bound from Theorem \ref{cotaghw} is not sharp. Let $q^s=4$, $d=3$, $m=2$ and $r=2$. Using Theorem \ref{paramRM}, we obtain $\dim \RM_{2}(2)=6$ and $\dim \RM_3(2)=10$. For $\PRM_d(1)$, we have that this code is a projective Reed-Solomon code, and therefore its dimension is $d+1=4$ in this case. Thus,
$$
    Y = \left\{(\alpha,\gamma): \begin{array}{c}
        0 \leq  \alpha \leq 2\\
        0\leq \gamma \leq  0 \\
        \end{array}\right\}=\{(0,0),(1,0),(2,0)\}.
$$
Note that $\gamma=0$ because $d<q^s$, see Remark \ref{remdmenorq}. With respect to the generalized Hamming weights of affine Reed-Muller codes, we can use Theorem \ref{ghwRM} to obtain:
\begin{table}[ht]
\centering
\begin{tabular}{||c||c|c||}
\hline
    $r$ & 1&2 \\
    \hline \hline
    $d_r(\RM_3(2))$ & 4&7\\
    $d_r(\RM_2(2))$ & 8&11\\
    \hline
\end{tabular}
\end{table}

Now we have all the values required to compute the bound from Theorem \ref{cotaghw}:
$$
\begin{aligned}
B_{0,0} &=d_2(\RM_2(2))=11, \\
B_{1,0} &=\max \{ d_2(\RM_3(2),d_1(\RM_2(2))\}+d_1(\PRM_1(1))=\max \{ 7,8\}+2=10,\\
B_{2,0} &=d_2(\RM_3(2)+d_2(\PRM_1(1))=7+3=10,
\end{aligned}
$$
where we have used that $d_r(\PRM_d(1))=q^s-d+r$ for $d\leq q^s$. Therefore, 
$$
d_2(\PRM_3(2))\geq \min_{(\alpha,\gamma)\in Y}B_{\alpha,\gamma}=\min \{11,10,10\}=10.
$$
However, according to \cite[Ex 7.5]{beelenGHWPRM}, the true value is $d_2(\PRM_3(2))=11$ in this case. 
\end{ex}

\begin{ex}
We mainly use Lemma \ref{lemasup} with the bound $d_r(\PRM_d(m))\leq d_r(\RM_{d-1}(m))=B_{0,0}$, but the part $d_r(\PRM_d(m))\leq q^s\cdot d_r(\PRM_d(m-1))$ is also useful in some cases. For example, for $q^s=3$, $d=1$, $r=2$, $m=2$, one can check that $\min_{(\alpha,\gamma)\in Y}B_{\alpha,\gamma}=12$. We cannot use the bound $d_r(\PRM_d(m))\leq d_r(\RM_{d-1}(m))$ because $\dim \RM_0(2)=1<2=r$. However, we have $r=2= \dim \PRM_1(1)$, and $12=q^s\cdot d_2(\PRM_1(1))$. 

In the previous case, the bound was useful because we simply could not use $d_r(\PRM_d(m))\leq d_r(\RM_{d-1}(m))$. Nevertheless, there are also cases in which the bound $d_r(\PRM_d(m))\leq q^s\cdot d_r(\PRM_d(m-1))$ is better than the bound $d_r(\PRM_d(m))\leq d_r(\RM_{d-1}(m))$. For instance, for $q^s=3$, $d=3$, $r=2$, $m=3$, we have $d_2(\RM_{2}(3))=15>12=3\cdot d_2(\PRM_3(2))$. In fact, one can also check in this case that $\min_{(\alpha,\gamma)\in Y}B_{\alpha,\gamma}=12$ and we can state that the bound from Theorem \ref{cotaghw} is sharp due to Lemma \ref{lemasup}.
\end{ex}

\subsection{A bound for the generalized Hamming weights of projective Reed-Muller codes over $\mathbb{P}^2$}
Even though the bound from Theorem \ref{cotaghw} is not hard to obtain by computer in general, it can be obtained in more efficient ways in some particular cases. In this subsection, we obtain the bound from Theorem \ref{cotaghw} in a more explicit way and requiring less values of $B_{\alpha,\gamma}$ in order to compute the minimum over $Y$, for the case $m=2$.

\begin{thm}\label{cotam2}
Let $1\leq d\leq 2(q^s-1)$, $2\leq r\leq \dim (\PRM_d(2))$, and $Y$ as in Theorem \ref{cotaghw}. 
\begin{enumerate}[wide, labelwidth=!, labelindent=0pt]
\item[(a)] If $d<q^s$, we consider $\alpha_0$ the smallest integer such that $d_{r-\alpha_0}(\RM_{d-1}(2))\leq d_{r}(\RM_d(2))$, $\mu_0=\max\{\alpha_0,r-\dim \RM_{d-1}(2)\}$ and $\lambda=\min\{d+1,r\}$. Then
$$
d_r(\PRM_d(2))\geq \min_{(\alpha,\gamma)\in Y}B_{\alpha,\gamma}=\begin{cases}
\min\{B_{0,0},H_{\alpha_0,0} \} &\text{ if } r\leq \dim \RM_{d-1}(2),\\
B_{\alpha_0,0} &\text{ if } r>\dim \RM_{d-1}(2),
\end{cases}
$$
where $B_{0,0}=d_r(\RM_{d-1}(2))$ and 
$$
H_{\alpha_0,0}=\begin{cases}
d_r(\RM_d(2))+q^s-d+\mu_0 &\text{ if } \alpha_0\leq \lambda,\\
d_{r-\lambda}(\RM_{d-1}(2))+q^s-d+\lambda &\text{ if }\alpha_0> \lambda.
\end{cases}
$$

\item[(b)] If $d\geq q^s$, we consider 
$$
E=\{\gamma \mid \max\{r-\dim \RM_d(2),0\}\leq \gamma \leq \min \{d-q^s+2,r\}\}.
$$
Let $i\leq \alpha_i\leq r$ be the smallest integer such that $d_{r-\alpha_i}(\RM_{d-1}(2))\leq d_{r-i}(\RM_d(2))$, for $i\in E$. We also consider $\mu_i=\max\{\alpha_i,r-\dim \RM_{d-1}(2)\}$ and $\lambda=\min \{q^s+1,r\}$. Then
\begin{equation}\label{igualdadBm2}
d_r(\PRM_d(2))\geq \min_{(\alpha,\gamma)\in Y}B_{\alpha,\gamma}=\begin{cases}
\min\{B_{0,0},\min_{\gamma\in E} \{H_{\alpha_\gamma,\gamma} \} \} &\text{ if } r\leq \dim \RM_{d-1}(2),\\
\min_{\gamma\in E} \{H_{\alpha_\gamma,\gamma} \} &\text{ if } r> \dim \RM_{d-1}(2),
\end{cases}
\end{equation}
where $B_{0,0}=d_r(\RM_{d-1}(2))$ and 
\begin{equation}\label{expH}
H_{\alpha_\gamma,\gamma}=\begin{cases}
d_{r}(\RM_d(2))+\mu_0 &\text{ if } \gamma=0, \alpha_\gamma \leq  \lambda, \\
d_{r-\lambda}(\RM_{d-1}(2))+\lambda &\text{ if } \gamma=0, \alpha_\gamma >  \lambda, \\
d_{r-\gamma}(\RM_d(2))+\max\{\mu_\gamma, 2q^s-d+\gamma-1\} &\text{ if } \gamma>0, \alpha_\gamma \leq  \lambda, \\
d_{r-\lambda}(\RM_{d-1}(2))+\max\{\lambda, 2q^s-d+\gamma-1\} &\text{ if } \gamma>0, \alpha_\gamma >  \lambda. \\
\end{cases}
\end{equation}
\end{enumerate}
Moreover, in both cases if $\min_{(\alpha,\gamma)\in Y}B_{\alpha,\gamma}=B_{0,0}=d_r(\RM_{d-1}(2))$ or $\min_{(\alpha,\gamma)\in Y}B_{\alpha,\gamma}=q^s\cdot \max\{q^s-d+r,r\},$ for $r\leq \dim\RM_{d-1}(2)$ or $r\leq \min\{d+1,q^s+1\}$, respectively, then the bound is sharp. 
\end{thm}
\begin{proof}
We prove the statement for $d\geq q^s$ and $r\leq \dim \RM_{d-1}(2)$ because the rest of the cases are argued in the same way (for instance, the case $d<q^s$ is analogous to the case of $\gamma=0$ with $d\geq q^s$, see Remark \ref{remdmenorq}). If we prove the equality in (\ref{igualdadBm2}), the rest follows from Theorem \ref{cotaghw} and Lemma \ref{lemasup}.

For $d\geq q^s$, we can rewrite the set $Y$ from Theorem \ref{cotaghw} for the case $m=2$ in the following way:
$$
    Y = \left\{(\alpha,\gamma): \begin{array}{c}
        \max\{r-\dim \RM_{d-1}(2),\gamma\} \leq  \alpha \leq \min \{q^s+1,r\}\\
        \max\{r-\dim\RM_d(2),0\}\leq \gamma \leq  \min\{d-q^s+2,r\} \\
        \end{array}\right\},
$$
where we have used the expressions for the dimension of the corresponding projective Reed-Solomon codes and we have written the conditions for $\alpha$ in terms of $\gamma$. For each $\gamma\in E$ we consider 
$$
Y_\gamma=\{ \alpha\mid (\alpha,\gamma)\in Y \}.
$$
This set is nonempty for each $\gamma\in E$ because $\gamma\leq \min\{d-q^s+2,r\}\leq \min\{q^s+1,r\}$, and we also have $r-\dim\RM_{d-1}(2)\leq \min\{q^s+1,r\}$ because of Corollary \ref{corodimRM}. We have that $Y=\bigcup_{\gamma\in E}Y_\gamma$, and therefore
$$
\min_{(\alpha,\gamma)\in Y}B_{\alpha,\gamma}=\min_{\gamma \in E}\{\min_{\alpha \in Y_\gamma}B_{\alpha,\gamma}\}.
$$

For $\gamma\in E$ fixed, we study now the behavior of $B_{\alpha,\gamma}$, as defined in (\ref{defcota}), as a function of $\alpha$. We first consider the case with $\gamma>0$. As we are assuming $2(q^s-1)\geq d\geq q^s$, we have that $\PRM_d(1)=\fqs^{q^s+1}$. Therefore, for $\alpha\geq \gamma>0$, we have 
$$
\max\{d_\alpha(\PRM_d(1)),d_\gamma(\PRM_{d-(q^s-1)}(1))\}=\max\{\alpha,2q^s-d+\gamma-1\}.
$$
For $\gamma>0$ fixed, $\max\{\alpha,2q^s-d+\gamma-1\}$ is a nondecreasing function of $\alpha$. For $1\leq \alpha < \alpha_\gamma$ we have that 
$$
\max\{d_{r-\gamma}(\RM_d(2)),d_{r-\alpha}(\RM_{d-1}(2))\}=d_{r-\alpha}(\RM_{d-1}(2)). 
$$
Using Theorem \ref{monotonia}, we see that this value decreases in at least one unit when we increase $\alpha$ in one unit, while $\max\{\alpha,2q^s-d+\gamma-1\}$ is going to increase by at most one unit. Therefore, $B_{\alpha,\gamma}$ is a nonincreasing function of $\alpha$, for $1\leq \alpha <\alpha_\gamma$.

For $\alpha\geq \alpha_\gamma$, we have 
$$
\max\{d_{r-\gamma}(\RM_d(2)),d_{r-\alpha}(\RM_{d-1}(2))\}=d_{r-\gamma}(\RM_d(2)),
$$
which is constant for a fixed $\gamma$. As $\max\{\alpha,2q^s-d+\gamma-1\}$ is nondecreasing, we obtain that $B_{\alpha,\gamma}$ is a nondecreasing function of $\alpha$, for $\alpha_\gamma\leq \alpha\leq r$. 

Therefore, we know the behavior of $B_{\alpha,\gamma}$ and we can obtain $\min_{\alpha \in Y_\gamma}B_{\alpha,\gamma}$. Indeed, for $1\leq\alpha <\alpha_\gamma$, $B_{\alpha,\gamma}$ is nonincreasing, and therefore, $\min_{\alpha\in Y_\gamma,\alpha<\alpha_\gamma}B_{\alpha,\gamma}$ is attained at the largest value of $\alpha \in Y_\gamma$ such that $\alpha<\alpha_\gamma$ (if any). And, as $B_{\alpha,\gamma}$ is nondecreasing for $\alpha_\gamma\leq \alpha\leq r$, we have that $\min_{\alpha\in Y_\gamma,\alpha\geq \alpha_\gamma}B_{\alpha,\gamma}$ is attained at the lowest value of $\alpha\in Y_\gamma$ such that $\alpha\geq \alpha_\gamma$. To obtain $\min_{\alpha \in Y_\gamma}B_{\alpha,\gamma}$, we just need to consider the minimum between the minimums in the case $1\leq \alpha<\alpha_\gamma$ and $\alpha_\gamma\leq\alpha\leq r$. We have two cases:
\begin{enumerate}
    \item[(a)] If $\alpha_\gamma\leq \lambda$: by definition, $\alpha_\gamma\geq \gamma$. Therefore, we have $\alpha_\gamma\in Y_\gamma$ if and only if $r-\dim\RM_{d-1}(2)\leq \alpha_\gamma$, which happens if and only if $\mu_\gamma=\alpha_\gamma$. If $\alpha_\gamma> r-\dim\RM_{d-1}(2)$, all the values of $Y_\gamma$ are larger than $\alpha_\gamma$, i.e., we are in the nondecreasing part of $B_{\alpha,\gamma}$, and the minimum of $B_{\alpha,\gamma}$ over $Y_\gamma$ is thus obtained at the lowest value of $Y_\gamma$, which is precisely $r-\dim \RM_{d-1}(2)=\mu_\gamma$ in this case. By the definition of $\alpha_\gamma$, we have
    $$
    B_{\mu_\gamma,\gamma}=d_{r-\gamma}(\RM_d(2))+\max \{\mu_\gamma,2q^s-d+\gamma-1\}.
    $$
    On the other hand, if $\alpha_\gamma\leq r-\dim \RM_{d-1}(2)$, i.e., $\alpha_\gamma=\mu_\gamma$, then the minimum in the nondecreasing part is $B_{\alpha_\gamma,\gamma}=B_{\mu_\gamma,\gamma}$ as above. If $\alpha_\gamma-1\not\in Y_\gamma$, this is the minimum over $Y_\gamma$. If $\alpha_\gamma-1\in Y_\gamma$, we have to also consider the minimum over the nonincreasing part, which, taking into account the definition of $\alpha_\gamma$, would be
    $$
    B_{\alpha_\gamma-1,\gamma}=d_{r-(\alpha_\gamma-1)}(\RM_{d-1}(2))+\max\{\alpha_\gamma-1,2q^s-d+\gamma-1\}.  
    $$
    If we compute the difference we obtain
    $$
    B_{\alpha_\gamma-1,\gamma}-B_{\alpha_\gamma,\gamma}\geq d_{r-(\alpha_\gamma-1)}(\RM_{d-1}(2))-d_{r-\gamma}(\RM_d(2))-1\geq 0
    $$
    because of the definition of $\alpha_\gamma$. Hence, in this case we also obtain $\min_{\alpha\in Y_\gamma} B_{\alpha,\gamma}=B_{\mu_\gamma,\gamma}$.
    \item[(b)] If $\alpha_\gamma >\lambda$: we have that all the values of $Y_\gamma$ are below $\alpha_\gamma$, i.e., we are in the nonincreasing part of $B_{\alpha,\gamma}$. The minimum is thus obtained at the maximum value in $Y_\gamma$, which is $\lambda$, and by the definition of $\alpha_\gamma$ we have
    $$
    B_{\lambda,\gamma}=d_{r-\lambda}(\RM_{d-1}(2))+\max \{\lambda,2q^s-d+\gamma-1\}.
    $$
\end{enumerate}
For the case with $\gamma=0$, we have $\max\{d_\alpha(\PRM_d(1)),d_\gamma(\PRM_{d-(q^s-1)}(1))\}=\alpha$. The previous argument also applies in this case for the $\alpha\in Y_0$ with $\alpha>0$ because this does not change the behaviour of $B_{\alpha,\gamma}$ (it only changes the exact expression of $B_{\alpha,\gamma}$, as seen in the statement). Therefore the minimum of $B_{\alpha,\gamma}$ in $Y_0\setminus \{0\}$ is either $B_{\mu_\gamma,\gamma}$ or $B_{\lambda,\gamma}$ as before. If $0\in Y_0$, which happens if and only if $r\leq \dim \RM_{d-1}(2)$, we also have to take into account for the minimum the bound $B_{0,0}=d_r(\RM_{d-1}(2))$. We complete the proof by noting that $H_{\alpha_\gamma,\gamma}$ is equal to $B_{\mu_\gamma,\gamma}$ or $B_{\lambda,\gamma}$, depending on where the minimum is attained in each case. 
\end{proof}


One can also express the previous result in a more explicit way by taking into account that
\begin{equation}\label{dimRM}
\dim \RM_d(2)=\begin{cases}
    \binom{d+2}{2} &\text{ if } d<q^s,\\
    \binom{d+2}{2}-2\binom{d-q^s+2}{2} &\text{ if } q^s\leq d\leq 2(q^s-1),\\
\end{cases}
\end{equation}
which can be proven from Theorem \ref{paramRM}. 

Although Theorem \ref{cotam2} looks more involved than Theorem \ref{cotaghw}, it can greatly simplify the procedure of computing the bound from Theorem \ref{cotaghw} for the case $m=2$, as we show in the next example. 

\begin{ex}\label{exm2}
Let $q^s=4$, $d=5$, $m=2$ and $r=5$. We are going to use Theorem \ref{cotam2} to obtain a bound for $d_4(\PRM_5(2))$. As we have $d\leq q^s$, we compute $\dim \RM_5(2)=15$ with (\ref{dimRM}) and we obtain
$$
E=\{\gamma\mid 0\leq \gamma\leq 3\}.
$$
Using Theorem \ref{ghwRM}, we can compute
\begin{table}[ht]
\centering
\begin{tabular}{||c||c|c|c|c|c||}
\hline
    $r$ & 1&2&3&4&5 \\
    \hline \hline
    $d_r(\RM_5(2))$ & 2&3&4&5&6\\
    $d_r(\RM_4(2))$ & 3&4&6&7&8\\
    \hline
\end{tabular}
\end{table}

From this table we obtain $\alpha_0=2$, $\alpha_1=\alpha_2=3$ and $\alpha_3=4$. For example, we check that $d_{5-3}(\RM_4(2))=4\leq 5=d_{5-2}(\RM_{5-1}(2))$, but $d_{5-2}(\RM_4(2))=6>5=d_{5-2}(\RM_{5-1}(2))$, which implies $\alpha_2=3$. Using (\ref{dimRM}) again, we can compute $\dim \RM_4(2)=13$, which implies that $r-\dim \RM_4(2)<0$ and $\mu_i=\alpha_i$ for $i=0,1,2,3$. We also have $\lambda=\min\{q^s+1,r\}=5$. Thus, $\alpha_i<\lambda$ for $i=0,1,2,3$. The only thing left to do is to compute $B_{0,0}$ and $H_{\alpha_\gamma,\gamma}$, for $\gamma \in E$. We obtain
$$
\begin{aligned}
B_{0,0}&=d_5(\RM_4(2))=8, \\
H_{2,0}&=d_5(\RM_5(2))+\mu_0=6+2=8,\\
H_{3,1}&=d_4(\RM_5(2))+\max\{\mu_1,2q^s-d+1-1\}=5+\max\{3,3\}=8,\\
H_{3,2}&=d_3(\RM_5(2))+\max\{\mu_2,2q^s-d+2-1\}=4+\max\{3,4\}=8,\\
H_{4,3}&=d_2(\RM_5(2))+\max\{\mu_3,2q^s-d+3-1\}=3+\max\{4,5\}=8.\\
\end{aligned}
$$
The minimum of these values is $8$, and we have $d_5(\PRM_5(2))\geq 8$. Furthermore, as the minimum is equal to $B_{0,0}$, by Theorem \ref{cotam2} (or Lemma \ref{lemasup}) we have the equality $d_5(\PRM_5(2))= 8$.

Note that if we want to use Theorem \ref{cotaghw} to obtain the bound for $d_5(\PRM_5(2))$, we would have to consider the minimum of $\abs{Y}$ terms. In this case, we have
$$
    Y = \left\{(\alpha,\gamma): \begin{array}{c}
        \gamma \leq  \alpha \leq \min\{q^s+1,r\}=5\\
       0\leq \gamma \leq  \min\{d-q^s+2,r\}=3 \\
        \end{array}\right\}.
$$
Thus, using Theorem \ref{cotaghw} we would have to consider the minimum of $\abs{Y}=18$ terms, while using Theorem \ref{cotam2} we only needed 5. 
\end{ex}

We now compare our results with the current literature on the generalized Hamming weights of projective Reed-Muller codes. Theorem \ref{cotaghw} provides a lower bound for the generalized Hamming weights for any degree $d$. Most of the literature, like \cite{beelenGHWPRM,singhalGHWPRM}, focuses on the case $d\leq q^s-1$. Thus, covering the case $d\geq q^s$ is one of the main contributions of this work with respect to the current literature. For the case $d\leq q^s-1$, the reference \cite{beelenGHWPRM} provides a conjectural formula for the generalized Hamming weights, which is also proved for several values of $r$.  Moreover, in \cite{singhalGHWPRM}, the authors prove the aforementioned conjecture for $q^s$ sufficiently large. In the examples we have checked, we have not found a counterexample to this conjecture using our bounds, which provides some computational support to this conjecture.

\subsection{A bound for the generalized Hamming weights of the subfield subcodes of projective Reed-Muller codes}
The subfield subcodes of Section \ref{secssc} are subcodes of projective Reed-Muller codes and therefore the generalized Hamming weights of projective Reed-Muller codes give lower bounds for the generalized Hamming weights of the corresponding subfield subcodes. However, as the dimension of the subfield subcode is usually much smaller than the dimension of the original code, this bound is not sharp for most of the generalized Hamming weights. Nevertheless, Theorem \ref{cotaghw} and Lemma \ref{lemasup} can be adapted for the subfield subcodes as well. 

\begin{cor}\label{ghwssc}
Let $\xi\in \fqs$ be a primitive element and $m>1$.  Let ${d_\lambda}=\lambda\frac{q^s-1}{q-1}$ for some $\lambda\in \{ 1,2,\dots,m(q-1)\}$ and $2\leq r\leq \dim(\PRM^\sigma_d(m))$. We consider
$$
    Y^\sigma = \left\{(\alpha,\gamma): \begin{array}{c}
        \max\{r-\dim \RM^\sigma_{{d_\lambda}-1}(m),0\} \leq  \alpha \leq \min \{ \dim \PRM^\sigma_{d_\lambda}(m-1),r\}\\
        \max\{r-\dim\RM^\sigma_{d_\lambda}(m),0\}\leq \gamma \leq  \min\{\dim \PRM^\sigma_{{d_\lambda}-(q^s-1)}(m-1),\alpha\} \\
        \end{array}\right\}.
$$
Then we have
$$
d_r(\PRM^\sigma_{d_\lambda}(m))\geq \min_{(\alpha,\gamma)\in Y^\sigma}B_{\alpha,\gamma}^\sigma,
$$
where 
$$
\begin{aligned}
B_{\alpha,\gamma}^\sigma:= \max(d_{r-\gamma}(\RM^\sigma_{d_\lambda}(m)),&d_{r-\alpha}(\RM^\sigma_{{d_\lambda}-1}(m))) \\
&+\max(d_{\alpha}(\PRM^\sigma_{d_\lambda}(m-1)),d_\gamma(\PRM^\sigma_{{d_\lambda}-(q^s-1)}(m-1))).
\end{aligned}
$$
\end{cor}
\begin{proof}
We consider the recursive construction from Corollary \ref{recursivessc}, and we apply an analogous reasoning to the one above Theorem \ref{cotaghw}. 
\end{proof}

\begin{cor}\label{ghwssc2}
Let $\xi\in \fqs$ be a primitive element and $m>1$. Let ${d_\lambda}=\lambda\frac{q^s-1}{q-1}$ for some $\lambda\in \{ 1,2,\dots,m(q-1)\}$ and $2\leq r\leq \max\{\dim \RM^\sigma_{{d_\lambda}-1}(m),\dim \PRM^\sigma_{d_\lambda}(m-1) \}$. Then
$$
d_r(\PRM^\sigma_{d_\lambda}(m))\leq \min\{ d_r(\RM^\sigma_{{d_\lambda}-1}(m)),q^s\cdot d_r(\PRM^\sigma_{d_\lambda}(m-1))\},
$$
where if $r> \dim\RM^\sigma_{{d_\lambda}-1}(m)) $ or $r> \dim \PRM^\sigma_{d_\lambda}(m-1)$ we do not consider the corresponding generalized Hamming weight in the minimum.
\end{cor}
\begin{proof}
We consider the recursive construction from Corollary \ref{recursivessc} and argue as in the proof of Lemma \ref{lemasup}. 
\end{proof}

With respect to the values that appear in these results, we have formulas for the dimension of the subfield subcodes of affine Reed-Muller codes \cite[Thm. 11]{galindolcd} and for the subfield subcodes of projective Reed-Muller codes over $\mathbb{P}^2$ \cite[Cor. 3.41]{sanjoseSSCPRM}. However, we do not have explicit results about the generalized Hamming weights of the subfield subcodes in the affine case or the projective case, besides the bounds given by the generalized Hamming weights of the affine or projective Reed-Muller codes. Therefore, although Corollary \ref{ghwssc} and Corollary \ref{ghwssc2} are of theoretical interest, their practical utility is more limited than that of Theorem \ref{cotaghw} and Lemma \ref{lemasup}. 

\subsection{Examples}\label{secex}
In this subsection we provide tables showing examples of how one can use the results in Section \ref{secghw} to determine the generalized Hamming weights of some projective Reed-Muller codes. In order to use the bound from Theorem \ref{cotaghw}, we use the true values for the minimum distances of projective Reed-Muller codes from Theorem \ref{paramPRM}. We will use Theorem \ref{cotaghw} to provide a lower bound, and we can use Lemma \ref{lemasup} to give an upper bound in some cases. Moreover, if we know the true value of a generalized Hamming weight $d_r(\PRM_d(m))$, for $r'\leq r$ we know that $d_{r'}(\PRM_d(m))\leq d_r(\PRM_d(m))-(r-r')$ due to the monotonicity from Theorem \ref{monotonia}, which provides another upper bound. With this we can obtain the tables that we present below. Note that we are only using the bounds that we have just stated, we are not considering the values obtained in other papers such as \cite{beelenGHWPRM}. For instance, by considering the values given in Example \ref{exnosharp}, we see that $d_2(\PRM_3(2))=11$ for $q^s=4$. Looking at Table \ref{q4m2}, we see that this also implies that $d_3(\PRM_3(2))=12$ by Theorem \ref{monotonia}, and we would therefore obtain all the generalized Hamming weights of $\PRM_3(2)$ for $q^s=4$. Nevertheless, as it is seen in the tables, we can still obtain many of the true values of the generalized Hamming weights for projective Reed-Muller codes. 

With respect to the notation, we will use dots when the generalized Hamming weights grow by one unit when increasing $r$. Unless stated otherwise, the value of the bound from Theorem \ref{cotaghw} coincides with the value of the generalized Hamming weight that appears in the table. When that does not happen (or if we do not know the true value), we will write the lower bound from Theorem \ref{cotaghw} and the best of the upper bounds that we have discussed above. We note that, for all the values that we omit by using dots, the bound from Theorem \ref{cotaghw} is sharp. 

As we stated previously, for $q^s=2$ we have checked by computer that we obtain the true values of $d_r(\PRM_d(m))$ for $m=2,3$. However, this case is not that important in the projective setting because, for $q^s=2$, we have $\mathbb{P}^m=\mathbb{A}^{m+1}\setminus \{(0,\dots,0)\}$. Therefore, the projective Reed-Muller codes over $\F_2$ in $\mathbb{P}^m$ are equal to the shortening of affine Reed-Muller codes over in $m+1$ variables over $\F_2$ at the point $(0,\dots,0)$. 

\begin{table}[ht]
\caption{Generalized Hamming weights for $q^s=3$, $m=2$.}
\label{q3m2}
\centering
\begin{tabular}{c|cccccccccccc}
$d\backslash r$ &2 &3 &4 &5 &6 &7 &8 &9 &10 &11 &12&13\\
\hline
1 &12 &13 & & & & & & & & & & \\
2 &8 &9 &11 &12 &13 & & & & & & & \\
3 &4-5 &6 &7 &8 &9 &10 &11 &12 &13 & & & \\
4 &3 &4 &5 &6 &7 &8 &9 &10 &11 &12 &13 & \\
\end{tabular}
\end{table}

\begin{table}[ht]
\caption{Generalized Hamming weights for $q^s=3$, $m=3$.}
\label{q3m3}
\centering

\begin{tabular}{c|ccccccccccccccccccccccccccc}
$d\backslash r$ &2 &3 &4 &5 &6 &7 &8 &9 &10 &11 &12 &13& $\cdots$&20\\
\hline
1 &36 &39 &40  \\
2 &23-24 &26 &27 &32-35 &35-36 &36-37 &38 &39 &40 \\
3 &12 &13-17 &18 &21 &22-23 &24 &25 &26 &27 &30-31 &31-32 &33 &$\cdots $& 40 \\
\end{tabular}

\begin{tabular}{c|ccccccccccccccccccccccccccccccccccccccccccc}
$d\backslash r$ &2 &3 &4 &5 &6 &7 &8 &9 &10 &11 &$\cdots$ &17&18& $\cdots$&29\\
\hline
4 &8 &9 &11-12 &12-14 &13-15 &15-16 &17 &18 &20 &21&$\cdots$& 27 &29&$\cdots$ &40 \\
\end{tabular}

\begin{tabular}{c|ccccccccccccc}
$d\backslash r$ &2 &3 &4 &5 &6 &7 &8 &9 &10 &11& $\cdots$&36\\
\hline
5 &4-5 &6 &7 &8 &9 &10-11 &11-12 &12-13 &13-14 &15 &$\cdots $&40 \\
\end{tabular}

\begin{tabular}{c|ccccccccccccc}
$d\backslash r$ &2& $\cdots$&39\\
\hline
6 &3 &$\cdots $&40 \\
\end{tabular}
\end{table}

\begin{table}[ht]
\caption{Generalized Hamming weights for $q^s=4$, $m=2$.}
\label{q4m2} 
\centering
\begin{tabular}{c|ccccccccccccccccccccc}
$d\backslash r$& 2 &3 &4 &5 &6 &7 &8 &9 &10 &11 & $\cdots$ &20 \\
\hline
1 &20 &21 \\
2 &15 &16 &19 &20 &21 \\
3 &10-11 &11-12 &14 &15 &16 &18 &19 &20 &21 \\
4 &5-7 &8 &9-10 &10-11 &12 &13 &14 &15 &16 &17&$\cdots$   \\
5 &4 &5-6 &7 &8 &9 &10 &11 &12 &13 &14& $\cdots$   \\
6 &3 &4 &5 &6 &7 &8 &9 &10 &11 &12& $\cdots$ &21  \\
\end{tabular}
\end{table}

\begin{table}[ht]
\caption{Generalized Hamming weights for $q^s=5$, $m=2$.}
\label{q5m2} 
\centering
\begin{tabular}{c|ccccccccccccccccccccc}
$d\backslash r$& 2 &3 &4 &5 &6 &7 &8 &9 &10 &11 &$\cdots$ &30 \\
\hline
1 &30 &31 \\
2 &24 &25 &29 &30 &31 \\
3 &18-19 &19-20 &23 &24 &25 &28 &29 &30 &31 \\
4 &12-14 &13-15 &17-18 &18-19 &19-20 &22 &23 &24 &25 &27 &$\cdots$& \\
5 &6-9 &10 &11-13 &12-14 &15 &16-17 &17-18 &18-19 &20 &21  &$\cdots $ & \\
6 &5 &6-8 &9 &10 &11-12 &12-13 &14 &15 &16 &17  &$\cdots $ & \\
7 &4 &5 &6-7 &8 &9 &10 &11 &12 &13 &14 &$\cdots $& \\
8 &3 &4 &5 &6 &7 &8 &9 &10 &11 &12 &$\cdots $&31 \\
\end{tabular}
\end{table}

We see in Table \ref{q3m2} that we recover all the generalized Hamming weights for $q^s=3$ and $m=2$, besides the one corresponding to $d=3$, $r=2$. For $q^s=3$ and $m=3$, we obtain most of the generalized Hamming weights (see Table \ref{q3m3}), although in some cases we cannot claim that we have the equality with the bound from Theorem \ref{cotaghw}. For example, for $d=5$, we obtain directly $30$ out of the $35$ generalized Hamming weights for $r>1$ (we recall that the bound is sharp for all the values represented with dots), and for the ones for which we only have bounds, the difference between the bound from Theorem \ref{cotaghw} and the true value is at most 1. 

For $q^s=4,5,$ and $m=2$, we show the values that we obtain in Table \ref{q4m2} and Table \ref{q5m2}. We see that for high values of $r$ we obtain the true values of the generalized Hamming weights, and for moderate values of $d$ and $r$, even in the cases where we are not able to state what the true value is, we have a small interval of possible values due to the bounds we are using. 

The tables can be further improved by using the following duality theorem for the generalized Hamming weights from \cite{weiGHW}.

\begin{thm}[(Duality)]\label{ghwdual}
Let $C$ be an $[n,k]$ code. Then
$$
\{d_r(C):1\leq r\leq k\}=\{1,2,\dots,n\}\setminus \{n+1-d_r(C^\perp):1\leq r\leq n-k\}.
$$
\end{thm}
The set $\{d_r(C):1\leq r\leq k\}$ is called the weight hierarchy of the code $C$. From Theorem \ref{ghwdual}, we see that the weight hierarchy of a code is completely determined by the weight hierarchy of its dual, and vice versa. As the dual of a projective Reed-Muller code is another projective Reed-Muller code (for $d\not \equiv 0\bmod q^s-1$), this gives us more restrictions for the possible values appearing in the previous tables. 

\begin{ex}\label{exghwdual}
Let $q^s=3$, $d=3$ and $m=3$. Looking at Table \ref{q3m3}, we see that for $r=3,6,11,12$, we do not know the exact value of the corresponding generalized Hamming weight of $\PRM_3(3)$. Using Theorem \ref{dualPRM}, we know that $\PRM_3^\perp(3)=\PRM_3(3)$. We are going to use Theorem \ref{ghwdual} to obtain the true value of more of the generalized Hamming weights of $\PRM_3(3)$. Using Theorem \ref{paramPRM}, we know that $d_1(\PRM_3(3))=9$. Therefore, by Theorem \ref{ghwdual}, $n+1-d_1(\PRM_3(3))=41-9=32$ is not in the weight hierarchy of $\PRM_3(3)$. Looking at Table \ref{q3m3}, this implies that $d_{12}(\PRM_3(3))=31$. By the monotonicity from Theorem \ref{monotonia}, this also implies that $d_{11}(\PRM_3(3))=30$.

If we consider $d_{10}(\PRM_3(3))=27$ from Table \ref{q3m3}, by Theorem \ref{ghwdual} we see that $41-27=14$ is not in the weight hierarchy of $\PRM_3(3)$. Considering the next weight $d_{11}(\PRM_3(3))=30$, we obtain that $11$ is not in the weight hierarchy. But by Theorem \ref{ghwdual}, as these are consecutive generalized Hamming weights and we have the monotonicity from Theorem \ref{monotonia}, this implies that $12$ and $13$ are contained in the weight hierarchy of $\PRM_3(3)$. Thus, $d_3(\PRM_3(3))=13$. Finally, by considering $d_4(\PRM_3(3))=18$, we obtain that $41-18=23$ does not belong to the weight hierarchy of $\PRM_3(3)$, and therefore $d_6(\PRM_3(3))=22$. Hence, we have been able to obtain the whole weight hierarchy of $\PRM_3(3)$ by using the bound from Theorem \ref{cotaghw} and the general properties of the generalized Hamming weights from \cite{weiGHW}. 
\end{ex}

The ideas from Example \ref{exghwdual} can be applied to improve the previous tables. For the case $q^s=3$, $m=2$, the only value that we did not have exactly was $d_2(\PRM_3(2))$, which must be equal to $4$ because $\PRM_3^\perp(2)=\PRM_1(2)$ and $d_1(\PRM_1(2))=9$, which means that $n+1-d_1(\PRM_1(2))=5$ is not in the weight hierarchy of $\PRM_3(2)$. For the rest of the tables, we show the improved values in Tables \ref{q3m3bis}, \ref{q4m2bis} and \ref{q5m2bis}. We note that we obtain almost all of the exact values of the generalized Hamming weights corresponding to the previous tables for $d\not\equiv 0\bmod q^s-1$. 
\begin{table}[ht]
\caption{Improved table of the generalized Hamming weights for $q^s=3$, $m=3$, with $d\not\equiv 0\bmod q^s-1$.}
\label{q3m3bis}
\centering

\begin{tabular}{c|ccccccccccccccccccccccccccc}
$d\backslash r$ &2 &3 &4 &5 &6 &7 &8 &9 &10 &11 &12 &13& $\cdots$&20\\
\hline
1 &36 &39 &40  \\
3 &12 &13 &18 &21 &22 &24 &25 &26 &27 &30 &31 &33 &$\cdots $& 40 \\
5 &4 &6 &7 &8 &9 &10 &11 &12 &13 &15 &16&17&$\cdots $&40 \\
\end{tabular}
\end{table}

\begin{table}[ht]
\caption{Improved table of the generalized Hamming weights for $q^s=4$, $m=2$, with $d\not\equiv 0\bmod q^s-1$.}
\label{q4m2bis} 
\centering
\begin{tabular}{c|ccccccccccccccccccccc}
$d\backslash r$& 2 &3 &4 &5 &6 &7 &8 &9 &10 &11 & $\cdots$ &18 \\
\hline
1 &20 &21 \\
2 &15 &16 &19 &20 &21 \\
4 &5 &8 &9 &11 &12 &13 &14 &15 &16 &17&$\cdots$   \\
5 &4 &5 &7 &8 &9 &10 &11 &12 &13 &14& $\cdots$ &21  \\
\end{tabular}
\end{table}

\begin{table}[ht]
\caption{Improved table of the generalized Hamming weights for $q^s=5$, $m=2$, with $d\not\equiv 0\bmod q^s-1$.}
\label{q5m2bis} 
\centering
\begin{tabular}{c|ccccccccccccccccccccc}
$d\backslash r$& 2 &3 &4 &5 &6 &7 &8 &9 &10 &11 &$\cdots$ &28 \\
\hline
1 &30 &31 \\
2 &24 &25 &29 &30 &31 \\
3 &18-19 &19-20 &23 &24 &25 &28 &29 &30 &31 \\
5 &6 &10 &11 &12-14 &15 &16 &18 &19 &20 &21  &$\cdots $ & \\
6 &5 &6 &9 &10 &11 &13 &14 &15 &16 &17  &$\cdots $ & \\
7 &4 &5 &6 &8 &9 &10 &11 &12 &13 &14 &$\cdots $& 31 \\
\end{tabular}
\end{table}
We also note that the generalized Hamming weights of projective Reed-Muller codes seem to achieve the generalized Singleton bound \ref{singletongeneralizada} in many cases. This can be checked in the tables by considering, for a fixed degree $d$, the smallest value $r^*$ such that $d_{r+1}(\PRM_d(m))-d_{r}(\PRM_d(m))=1$ for all $r\geq r^*$. All the generalized Hamming weights $d_r(\PRM_d(m))$ with $r\geq r^*$ achieve the generalized Singleton bound. For instance, in Table \ref{q4m2bis}, for $d=4$ we have that $d_r(\PRM_4(2))$ achieves the generalized Singleton bound for $r\geq 5$ (in this case, $n-k+r=21-15+r=6+r$). 

For the cases where we have $d\equiv 0 \bmod q^s-1$, the weight hierarchy of $\PRM_d(m)$ gives information about the weight hierarchy of $\PRM_d^\perp(m)=\PRM_{d^\perp}(m)+\langle (1,\dots,1)\rangle$. In particular, if we consider $r^*$ as before, then
$$
d_1(\PRM_d^\perp(m))=d_1(\PRM_{d^\perp}(m)+\langle (1,\dots,1)\rangle)=n+2-d_{r^*}(\PRM_d(m)).
$$
As we are able to obtain exactly the value of the generalized Hamming weights for large values of $r$ in all the cases that we have checked, we can obtain the minimum distance of $\PRM^\perp_d(m)$ for the case $d\equiv 0 \bmod q^s-1$. The minimum distances of these codes are the only basic parameters missing in \cite{sorensen} for the family of projective Reed-Muller codes and their duals. Clearly, by using more of the generalized Hamming weights of $\PRM_d(m)$ we can obtain (or at least bound) more generalized Hamming weights of $\PRM^\perp_{d}(m)$ besides the minimum distance in the case $d\equiv 0\bmod q^s-1$. 



\bibliographystyle{abbrv}

\end{document}